\documentclass[aip,rsi,preprint,groupedaddress,superscriptaddress,floatfix]{revtex4}

\usepackage{hyperref,color}
\usepackage{amsmath,amsfonts,amsthm,graphicx,harpoon,amssymb}
\usepackage[center]{titlesec}
\titleformat{\section}[hang]{\Large\bfseries\filcenter}{}{1em}{}
\titleformat{\subsection}[hang]{\bfseries}{}{1em}{}

\newtheorem{conjecture}{Conjecture}

\newtheorem{definition}{Definition}

\newtheorem{lemma}{Lemma}
\newtheorem{theorem}{Theorem}
\newtheorem{proposition}{Proposition}

\newcommand{\bra}[1]{\langle{#1}|}
\newcommand{\ket}[1]{|{#1}\rangle}

\newcommand{\abs}[1]{\left\lvert {#1} \right\rvert}
\def\bma{\begin{bmatrix}}
\def\ema{\end{bmatrix}}
\def\rank{\mathop{\rm rank}}

\def\diag{\mathop{\rm diag}}
\def\tr{{\rm Tr}}

\def\dg{\dagger}

\def\dim{\mathop{\rm Dim}}

\def\ox{\otimes}

\def\l{\lambda}
\def\m{\mu}

\def\p{\pi}
\def\r{\rho}

\def\ps{\psi}

\def\G{\Gamma}

\def\L{\Lambda}

\newcommand{\nc}{\newcommand}

 \nc{\bbA}{\mathbb{A}} \nc{\bbB}{\mathbb{B}} \nc{\bbC}{\mathbb{C}}
 \nc{\bbD}{\mathbb{D}} \nc{\bbE}{\mathbb{E}} \nc{\bbF}{\mathbb{F}}
 \nc{\bbG}{\mathbb{G}} \nc{\bbH}{\mathbb{H}} \nc{\bbI}{\mathbb{I}}
 \nc{\bbJ}{\mathbb{J}} \nc{\bbK}{\mathbb{K}} \nc{\bbL}{\mathbb{L}}
 \nc{\bbM}{\mathbb{M}} \nc{\bbN}{\mathbb{N}} \nc{\bbO}{\mathbb{O}}
 \nc{\bbP}{\mathbb{P}} \nc{\bbQ}{\mathbb{Q}} \nc{\bbR}{\mathbb{R}}
 \nc{\bbS}{\mathbb{S}} \nc{\bbT}{\mathbb{T}} \nc{\bbU}{\mathbb{U}}
 \nc{\bbV}{\mathbb{V}} \nc{\bbW}{\mathbb{W}} \nc{\bbX}{\mathbb{X}}
 \nc{\bbZ}{\mathbb{Z}}

 \nc{\bA}{{\bf A}} \nc{\bB}{{\bf B}} \nc{\bC}{{\bf C}}
 \nc{\bD}{{\bf D}} \nc{\bE}{{\bf E}} \nc{\bF}{{\bf F}}
 \nc{\bG}{{\bf G}} \nc{\bH}{{\bf H}} \nc{\bI}{{\bf I}}
 \nc{\bJ}{{\bf J}} \nc{\bK}{{\bf K}} \nc{\bL}{{\bf L}}
 \nc{\bM}{{\bf M}} \nc{\bN}{{\bf N}} \nc{\bO}{{\bf O}}
 \nc{\bP}{{\bf P}} \nc{\bQ}{{\bf Q}} \nc{\bR}{{\bf R}}
 \nc{\bS}{{\bf S}} \nc{\bT}{{\bf T}} \nc{\bU}{{\bf U}}
 \nc{\bV}{{\bf V}} \nc{\bW}{{\bf W}} \nc{\bX}{{\bf X}}
 \nc{\bZ}{{\bf Z}}

\nc{\cA}{{\cal A}} \nc{\cB}{{\cal B}} \nc{\cC}{{\cal C}}
\nc{\cD}{{\cal D}} \nc{\cE}{{\cal E}} \nc{\cF}{{\cal F}}
\nc{\cG}{{\cal G}} \nc{\cH}{{\cal H}} \nc{\cI}{{\cal I}}
\nc{\cJ}{{\cal J}} \nc{\cK}{{\cal K}} \nc{\cL}{{\cal L}}
\nc{\cM}{{\cal M}} \nc{\cN}{{\cal N}} \nc{\cO}{{\cal O}}
\nc{\cP}{{\cal P}} \nc{\cQ}{{\cal Q}} \nc{\cR}{{\cal R}}
\nc{\cS}{{\cal S}} \nc{\cT}{{\cal T}} \nc{\cU}{{\cal U}}
\nc{\cV}{{\cal V}} \nc{\cW}{{\cal W}} \nc{\cX}{{\cal X}}
\nc{\cZ}{{\cal Z}}

\nc{\cnn}{{\cal NN}}

\begin{document}

\title{On a matrix inequality related to the distillability problem}

\author{Yi Shen}\email[]{sy1609134@buaa.edu.cn}
\affiliation{School of Mathematics and Systems Science, Beihang University, Beijing 100191, China}

\author{Lin Chen}\email[]{linchen@buaa.edu.cn (corresponding author)}
\affiliation{School of Mathematics and Systems Science, Beihang University, Beijing 100191, China}
\affiliation{International Research Institute for Multidisciplinary Science, Beihang University, Beijing 100191, China}

\begin{abstract}
We investigate the distillability problem in quantum information in $\bbC^d\ox\bbC^d$. A special case of the problem has been reduced to proving a matrix inequality when $d=4$. We investigate the inequality for two families of non-normal matrices.
We prove the inequality for the first family with $d=4$ and two special cases of the second family with $d\ge4$. We also prove the inequality for all normal matrices with $d>4$.
\end{abstract}

\date{\today}

\maketitle

%\tableofcontents

\section{Introduction}
\label{sec:int}

We refer to $\mathbb{C}^{n\times n}$ as the set of $n\times n$ matrices with entries in the complex field, and $\mathbb{H}^{n\times n}$ as the set of $n\times n$ Hermitian matrices. 
Let $I_n$ be the identity matrix in $\mathbb{C}^{n\times n}$. We shall omit the subscript of the identity matrix when it is clear in the paper. Let $A\in\bbC^{n\times n}$ and $B\in\bbC^{m\times m}$. The Kronecker sum of $A$ and $B$ is defined as $A\ox I_m+I_n\ox B$, see more facts in \cite[section 7.2]{bernstein2009matrix}. 
Ref. \cite{pphh2010} has presented the following conjecture on the Kronecker sum when $A$ and $B$ have the same size. 

\begin{conjecture}
\label{cj:main}
Let $A,B,I\in\bbC^{d\times d}$, $d\geq4$, and the matrix
\begin{equation}                                                          
\begin{aligned}
\label{eq:x}
X=A\otimes I+I\otimes B
\end{aligned}
\end{equation}
where
\begin{equation}                                                          
\begin{aligned}
\label{eq:tra}
\tr A=\tr B=0,\quad \tr A^\dagger A+\tr B^\dagger B=\frac{1}{d}.
\end{aligned}
\end{equation}
Let $\sigma_1,\cdots,\sigma_{d^2}$ be the singular values of $X$ in the descending order.
Then
\begin{equation}
\begin{aligned}
\label{eq3}
\sup_{X}
~(\sigma_1^2+\sigma_2^2)\leq\frac{1}{2}.
\end{aligned}
\end{equation}
\qed
\end{conjecture}
The condition $d\ge4$ is essential, because we will show in Lemma \ref{le:d=3} that Conjecture \ref{cj:main} fails for $d=3$. It has been shown \cite{pphh2010} that Conjecture \ref{cj:main} for $d=4$ is a special case of the distillability problem. We will mathematically explain the special case in  Appendix \ref{app:dis}, due to the heavy terminologies from quantum physics. The distillability problem has been a main open problem in quantum information \cite{dss00} for a long time. It lies at the heart of entanglement theory \cite{Terhal2000A, Poon2015Preservers, Cariello2016A} and is related to the separability problem extensively studied by the mathematics community recently \cite{Hou2013Linear, Alfsen2012Finding, Woerdeman2004The}. We briefly introduce the physical motivation of distillability problem, and will give more details in Appendix \ref{app:dis}. In quantum physics, a quantum state is mathematically described by a positive semidefinite matrix.
The state is pure when it has rank one, otherwise the state is mixed. Pure entangled states play an essential role in most quantum-information tasks such as quantum computation. Nevertheless, there is no pure state in nature due to the unavoidable decoherence between the state and environment. So asymptotically converting initially bipartite entangled mixed states into bipartite pure entangled states under local operations and classical communications (LOCC) is a key step in quantum information processing. The distillability problem \cite{dss00,dcl00} asks whether the above-mentioned conversion succeeds for any mixed states. There have been some attempts to the problem in the past years  \cite{pphh2010, dss00,dcl00,vd06,br03,cc08,hhh97,rains1999,cd11jpa, cd16pra}. 

We return to Conjecture \ref{cj:main}. Although it is only a special case of the distillability problem, it has been an open problem for years.
Evidently, the matrix $X$ in \eqref{eq:x} is normal if and only if $A$ and $B$ in \eqref{eq:x} are both normal. Ref. \cite{pphh2010} has shown that 
\begin{lemma}
\label{le:d=4}
Let $\cN_d$ be the subset of normal matrices $X$ in Conjecture \ref{cj:main}. Then Eq.
\eqref{eq3} holds when $X\in\cN_4$.
\qed
\end{lemma}
We shall review the proof of Lemma \ref{le:d=4} in appendix \ref{app:lemma1}. The remaining work on Conjecture \ref{cj:main} is to prove it when $X$ is non-normal. It turns out to be a hard problem and there is no progress so far, as far as we know. In this paper we investigate Conjecture \ref{cj:main} in terms of two families of non-normal matrices. They are respectively constructed in Definition \ref{df:cp} and \ref{df:nonnormal}.
We prove Conjecture \ref{cj:main} for the first family of non-normal $X$
in Theorem \ref{thm:nonnormal}, based on Proposition \ref{pp:nonnormal1+nonnormal1} and \ref{pp:nonnormal+normal}. For the second family of non-normal $X$, we prove two special cases of Conjecture \ref{cj:main}. The cases respectively occur when the matrix $A$ in $X$ has rank one in Lemma \ref{le:ra1} and when $d\ge5$ in Lemma \ref{le:ada}. Our results carry out the first step of proving Conjecture \ref{cj:main} for non-normal matrices, and thus the distillability problem in quantum information. We shall also prove that Conjecture \ref{cj:main} holds for normal $X$ with $d>4$ in Lemma \ref{le:d>4}. Combining with Lemma \ref{le:d=4}, we obtain that Conjecture \ref{cj:main} holds for any normal $X$.

The rest of this paper is organized as follows. We introduce some notations and preliminary results in linear algebra in Sec. \ref{sec:pre}. We investigate Conjecture \ref{cj:main} for two families of non-normal matrices in Sec. \ref{sec:nonnormal-cl} 
and \ref{sec:nonnormal-sy}, respectively.
We further prove Conjecture \ref{cj:main} for normal matrices in Sec. \ref{sec:normal}.

\section{\label{sec:Pre}Preliminaries}
\label{sec:pre}
We shall denote $A^\dag$ as the conjugate transpose of matrix $A$. Let $\sigma(A)$ be the spectrum of matrix $A$, $\l_i(A)$ be an eigenvalue of $A$, and $A_{(i,j)}$ be the $(i,j)$ entry of $A$.
We post some lemmas used in the following sections. The following lemma is clear.

\begin{lemma}
\label{le:cj}
The following four statements are equivalent.

(i) Conjecture \ref{cj:main} holds.

(ii) Conjecture \ref{cj:main} holds when $X$ is replaced by $X^T$, $X^*$ or $X^\dg$.

(iii) Conjecture \ref{cj:main} holds   when $X$ is replaced by $(U\ox V)X(U^\dg\ox V^\dg)$ with any unitary matrices $U$ and $V$. 

(iv) Conjecture \ref{cj:main} holds when $X$ is replaced by $I\ox A+B\ox I$.
\qed
\end{lemma}

Remark: Using statement (iii) we can assume that $A$ and $B$ in Conjecture \ref{cj:main} are both upper-triangular. In particular, we can assume that they are diagonal if and only if they are normal.

\begin{lemma}\cite[Fact 4.10.16.]{bernstein2009matrix}(Gershgorin circle theorem)
\label{le:cirth}
Let $A\in \mathbb{C}^{n\times n}$. Then,
\begin{equation}
\begin{aligned}
\label{eq:spec1}
\sigma(A)\subset G(A)=\bigcup_{i=1}^n\big\{s\in \mathbb{C}: \abs{s-A_{(i,i)}}\leq\sum_{j=1\atop j\neq i}^n\abs{A_{(i,j)}} \big\},
\end{aligned}
\end{equation}
and a corollary is
\begin{equation}
\begin{aligned}
\label{eq:spec11}
\sigma(A)\subset \bigcup_{i=1}^n\big\{s\in \mathbb{C}: \abs{s}\leq\sum_{j=1}^n\abs{A_{(i,j)}} \big\}.
\end{aligned}
\end{equation} 

Remark: Let $R_i=\sum\limits_{j=1\atop j\neq i}^n\abs{A_{(i,j)}}$ and $D(A_{(i,i)},R_i)$ be the closed disc centered at $A_{(i,i)}$ with radius $R_i$. Every eigenvalue of $A$ lies within at least one of the Gershgorin discs $D(A_{(i,i)},R_i)$.
\qed 
\end{lemma}

\begin{lemma}\cite[Fact 4.10.21.]{bernstein2009matrix}(Brauer theorem)
\label{le:ovalth}
Let $A\in \mathbb{C}^{n\times n}$. Then,
\begin{equation}
\begin{aligned}
\label{eq:spec2}
\sigma(A)\subset \bigcup_{{i,j=1\atop i\neq j}}^n\big\{s\in \mathbb{C}: \abs{s-A_{(i,i)}}\abs{s-A_{(j,j)}}\leq\sum_{k=1\atop k\neq i}^n\abs{A_{(i,k)}}\sum_{k=1\atop k\neq j}^n\abs{A_{(j,k)}} \big\}.
\end{aligned}
\end{equation}

Remark: The eigenvalues of $A$ lie in the union of $n(n-1)/2$ ovals of Cassini which is contained in the union of Gershgorin discs \eqref{eq:spec1}. Hence, Brauer theorem is stronger than Gershgorin circle theorem.
\qed
\end{lemma}

\begin{lemma}\cite[Corollary 4.3.15.]{hj85}
\label{le:weylcor}
Let $A,B\in \mathbb{H}^{n\times n}$. Let the eigenvalues of $A,B$ be in the increasing order, that is $\lambda_1\leq\lambda_2\leq\cdots\leq\lambda_n$. Then for all $i=1,\cdots n,$ we have
\begin{equation}
\begin{aligned}
\label{eq:eivaineq1}
\lambda_i(A)+\lambda_1(B)\leq\lambda_i(A+B)\leq\lambda_i(A)+\lambda_n(B).
\end{aligned}
\end{equation}
\qed
\end{lemma}

\section{Conjecture \ref{cj:main} with non-normal matrices $X$: family 1}
\label{sec:nonnormal-cl}

In this section we prove Conjecture \ref{cj:main} with a family of non-normal matrices $X$ in Definition \ref{df:cp}. We will construct our main result in Theorem \ref{thm:nonnormal}, followed by two preliminary facts i.e. Proposition \ref{pp:nonnormal1+nonnormal1} and \ref{pp:nonnormal+normal}. 

\begin{definition}
\label{df:cp}	
Let $\cP$ be the subset of matrices $X$ with $d=4$, 
such that $A,B$ are normal or have the expressions $A=\bma
0&a_1&0&0\\
a_2&0&0&0\\
0&0&0&a_3\\
0&0&a_4&0\\
\ema$ and $B=
\bma
0&b_1&0&0\\
b_2&0&0&0\\
0&0&0&b_3\\
0&0&b_4&0\\
\ema$. 
\qed
\end{definition}
Since the complex numbers $a_i,b_j$ satisfy \eqref{eq:tra}, $X\in\cP$ may be normal or non-normal. We present the main result of this section as follows. 
\begin{theorem}
\label{thm:nonnormal}	
Eq. \eqref{eq3} holds when $X\in\cP$.	
\end{theorem}
\begin{proof}
If $A=
\bma
0&a_1&0&0\\
a_2&0&0&0\\
0&0&0&a_3\\
0&0&a_4&0\\
\ema$ and $B=
\bma
0&b_1&0&0\\
b_2&0&0&0\\
0&0&0&b_3\\
0&0&b_4&0\\
\ema$ then the assertion follows from Proposition \ref{pp:nonnormal1+nonnormal1}. If one of $A$ and $B$ is normal then we may assume that it is diagonal by Lemma \ref{le:cj} (iii). So the assertion follows from Proposition \ref{pp:nonnormal+normal} and the switch of $A,B$ (if any). If $A$ and $B$ are both normal then the assertion follows from Lemma  \ref{le:d=4}. This completes the proof.
\end{proof}

\begin{proposition}
\label{pp:nonnormal1+nonnormal1}
Eq. \eqref{eq3} holds when $X\in\cP$, $A=
\bma
0&a_1&0&0\\
a_2&0&0&0\\
0&0&0&a_3\\
0&0&a_4&0\\
\ema$ and $B=
\bma
0&b_1&0&0\\
b_2&0&0&0\\
0&0&0&b_3\\
0&0&b_4&0\\
\ema$.	
\end{proposition}
\begin{proof}
By computing one can show that $X^\dg X=Y_1\oplus Y_2$. We can partion $Y_1$ like $
        \left[
        \begin{array}{c|c}
        Z_1 & Z_2 \\ \hline
        Z_2^\dg & Z_4
        \end{array}
        \right]$,
where $Z_1=\diag(\abs{a_2}^2+\abs{b_2}^2,\abs{a_2}^2+\abs{b_1}^2,\abs{a_2}^2+\abs{b_4}^2,\abs{a_2}^2+\abs{b_3}^2)$, $Z_2=
\bma
0&a_1b_2^*+b_1a_2^*&0&0\\
a_1b_1^*+b_2a_2^*&0&0&0\\
0&0&0&a_1b_4^*+b_3a_2^*\\
0&0&a_1b_3^*+b_4a_2^*&0\\
\ema$ and $Z_4=\diag(\abs{a_1}^2+\abs{b_2}^2,\abs{a_1}^2+\abs{b_1}^2,\abs{a_1}^2+\abs{b_4}^2,\abs{a_1}^2+\abs{b_3}^2)$. One can calculate the eigenpolynomial of $Y_1$ as follow
\begin{equation}
\begin{aligned}
\label{eq:eigenpolyx2}
\abs{\lambda I-Y_1}&=[\big(\lambda-(\abs{a_1}^2+\abs{b_2}^2)\big)\big(\lambda-(\abs{a_2}^2+\abs{b_1}^2)\big)-\abs{a_1b_1^*+a_2^*b_2}^2]\cdot\\
&[\big(\lambda-(\abs{a_1}^2+\abs{b_1}^2)\big)\big(\lambda-(\abs{a_2}^2+\abs{b_2}^2)\big)-\abs{a_1b_2^*+a_2^*b_1}^2]\cdot\\
&[\big(\lambda-(\abs{a_1}^2+\abs{b_4}^2)\big)\big(\lambda-(\abs{a_2}^2+\abs{b_3}^2)\big)-\abs{a_1b_3^*+a_2^*b_4}^2]\cdot\\
&[\big(\lambda-(\abs{a_1}^2+\abs{b_3}^2)\big)\big(\lambda-(\abs{a_2}^2+\abs{b_4}^2)\big)-\abs{a_1b_4^*+a_2^*b_3}^2].
\end{aligned}
\end{equation}
We claim that the larger root of each quadratic polynomial in each line of \eqref{eq:eigenpolyx2} isn't greater than $\frac{1}{4}$. Substituting $\lambda=\sum\limits_{i=1}^4(\abs{a_i}^2+\abs{b_i}^2)$ into the first line of \eqref{eq:eigenpolyx2}, we have $(\abs{a_2}^2+\abs{a_3}^2+\abs{a_4}^2+\abs{b_1}^2+\abs{b_3}^2+\abs{b_4}^2)(\abs{a_1}^2+\abs{a_3}^2+\abs{a_4}^2+\abs{b_2}^2+\abs{b_3}^2+\abs{b_4}^2)-\abs{a_1b_1^*+a_2^*b_2}^2\geq0$, since $\abs{a_1b_1^*+a_2^*b_2}^2\leq\abs{a_1}^2\abs{b_1}^2+\abs{a_2}^2\abs{b_2}^2+\abs{a_1}^2\abs{a_2}^2+\abs{b_1}^2\abs{b_2}^2$. We can make the same conclusion in the same way to substitute $\lambda=\sum\limits_{i=1}^4(\abs{a_i}^2+\abs{b_i}^2)$ into other lines of \eqref{eq:eigenpolyx2}. Hence, the larger root of each quadratic polynomial in each line of \eqref{eq:eigenpolyx2} isn't greater than $\sum\limits_{i=1}^4(\abs{a_i}^2+\abs{b_i}^2)=\frac{1}{4}$. So any eigenvalue of $Y_1$ isn't greater than $\frac{1}{4}$. One can show that $Y_2$ can be evolved from $Y_1$ by replacing $a_1$ with $a_3$ and replacing $a_2$ with $a_4$ in $Y_1$. Hence, we can get a similar formulation of $\abs{\lambda I-Y_2}$ like \eqref{eq:eigenpolyx2} by replacing $a_1$ with $a_3$ and replacing $a_2$ with $a_4$ in \eqref{eq:eigenpolyx2}. In the same way, we conclude that any eigenvalue of $Y_2$ isn't greater than $\frac{1}{4}$. Since $X^\dg X=Y_1 \oplus Y_2$, the sum of the largest two eigenvalues of $X^\dg X$ is at most $\frac{1}{2}$. This completes the proof.
\end{proof}

We proceed with the proof of Proposition \ref{pp:nonnormal+normal}.
For this purpose we need the following two preliminary results. The first result is known as one of the basic inequalities.
\begin{lemma}
\label{le:bi}
If $a,b,x,y\in\bbR$ then $ab(x+y)^2\le (a+b)(ax^2+by^2)$.	
\end{lemma}

\begin{lemma}
\label{le:nonnormal+normal}
Suppose $a_1,a_2,b_1,b_2$ are nonnegative real numbers and $a_1^2+a_2^2+b_1^2+b_2^2=1/4$. Then 
\begin{equation}
\label{eq:a12}
\sqrt{(a_1^2-a_2^2)^2+4b_1^2(a_1+a_2)^2}
+
\sqrt{(a_1^2-a_2^2)^2+4b_2^2(a_1+a_2)^2}
\le1/2.
\end{equation}
\end{lemma}
\begin{proof}
Using the basic inequality $x+y\le \sqrt{2(x^2+y^2)}$ for any real $x,y$, we obtain that the lhs of \eqref{eq:a12} is upper bounded by  
\begin{eqnarray}
&&
\sqrt{2\bigg(2(a_1^2-a_2^2)^2+4(b_1^2+b_2^2)(a_1+a_2)^2\bigg)}	
\notag\\
&=&
\sqrt{2(a_1+a_2)^2\bigg(1-2(a_1+a_2)^2\bigg)}	
\notag\\
&\le & 1/2.
\end{eqnarray}
The equality follows from the equation $a_1^2+a_2^2+b_1^2+b_2^2=1/4$. 
This completes the proof.	
\end{proof}

\begin{proposition}
\label{pp:nonnormal+normal}
Eq. \eqref{eq3} holds when $X\in\cP$, $A=
\bma
0&a_1&0&0\\
a_2&0&0&0\\
0&0&0&a_3\\
0&0&a_4&0\\
\ema$ and $B=\diag(b_1,b_2,b_3,b_4)$.	
\end{proposition}
\begin{proof}
Let $X=A\ox I+I\ox B$ in \eqref{eq:x}. Since $A$ and $B$ satisfy \eqref{eq:tra}, we have
\begin{eqnarray}
\label{eq:non1}
\sum^4_{i=1} b_i&=&0,
\\	
\sum^4_{j=1} (\abs{a_j}^2+\abs{b_j}^2)&=&1/4.
\label{eq:non2}
\end{eqnarray}
By computation one can show that $X^\dag X$ has the same eigenvalues with that of $\oplus^4_{j=1} (Y_j \oplus Z_j)$ where $Y_j$ and $Z_j$ are order-2 submatrices such that
\begin{equation}
\label{eq:yj}
Y_j=
\bma
\abs{a_2}^2+\abs{b_j}^2 & b_j^*a_1+a_2^*b_j
\\
b_ja_1^*+a_2b_j^* & \abs{a_1}^2+\abs{b_j}^2 
\ema,	
\end{equation}
and
\begin{equation}
\label{eq:zj}
Z_j=
\bma
\abs{a_4}^2+\abs{b_j}^2 & b_j^*a_3+a_4^*b_j
\\
b_ja_3^*+a_4b_j^* & \abs{a_3}^2+\abs{b_j}^2 
\ema.	
\end{equation}
Let $\l$ and $\m$ be two arbitrary eigenvalues of $X^\dg X$. Then proving Conjecture \ref{cj:main} is equivalent to proving $\l+\m\le1/2$. We investigate five cases for $\l$ and $\m$.

Case 1. $\l$ and $\m$ are the eigenvalues of the same $Y_j$ or $Z_j$. Eqs. \eqref{eq:non2} and \eqref{eq:yj}
imply that $\l+\m=\tr Y_j \le 1/2$. Eqs. \eqref{eq:non2} and \eqref{eq:zj}
imply that $\l+\m=\tr Z_j \le 1/2$. So Conjecture \ref{cj:main} holds.

Case 2. $\l$ and $\m$ are the eigenvalues of different $Y_j$'s. Without loss of generality we can assume that $\l$ is the maximum eigenvalue of  $Y_1$, and $\m$ is the maximum eigenvalue of  $Y_2$. By computation one can obtain
\begin{eqnarray}
\label{eq:non3}
\l &=& {1\over2}\bigg( \abs{a_1}^2+\abs{a_2}^2+2\abs{b_1}^2+\sqrt{(\abs{a_1}^2-\abs{a_2}^2)^2+4\abs{b_1a_2^*+a_1b_1^*}^2}\bigg),
\\
\label{eq:non4}
\m &=& {1\over2}\bigg( \abs{a_1}^2+\abs{a_2}^2+2\abs{b_2}^2+\sqrt{(\abs{a_1}^2-\abs{a_2}^2)^2+4\abs{b_2a_2^*+a_1b_2^*}^2}\bigg).
\end{eqnarray} 
So $\l+\m$ is upper bounded by the sum of the rhs of \eqref{eq:non3} and \eqref{eq:non4}, in which any $a_i$ and $b_j$ are replaced by $\abs{a_i}$ and $\abs{b_j}$, respectively, and $a_3,a_4,b_3,b_4$ equal zero. Using Lemma \ref{le:nonnormal+normal} and \eqref{eq:non2}, we have $\l+\m\le1/2$. So Conjecture \ref{cj:main} holds.

Case 3. $\l$ and $\m$ are the eigenvalues of different $Z_j$'s. We can prove Conjecture \ref{cj:main} by following the proof in Case 2, except that we switch $a_1$ and $a_3$, and switch $a_2$ and $a_4$ at the same time.

Case 4. $\l$ is the eigenvalue of some $Y_j$, $\m$ is the eigenvalue of some $Z_k$, and $j\ne k$.
Without loss of generality we may assume that $j=1$ and $k=2$. By computation one can show that
\begin{eqnarray}
\label{eq:case4}
\l+\m
&=&
{1\over2}\bigg( 
\abs{a_1}^2+\abs{a_2}^2+\abs{a_3}^2+\abs{a_4}^2+2(\abs{b_1}^2+\abs{b_2}^2)
\notag\\
&+&\sqrt{(\abs{a_1}^2-\abs{a_2}^2)^2+4\abs{b_1a_2^*+a_1b_1^*}^2}
+\sqrt{(\abs{a_3}^2-\abs{a_4}^2)^2+4\abs{b_2a_4^*+a_3b_2^*}^2}
\bigg)
\notag\\
&\le &
{1\over2}\bigg( 
\abs{a_1}^2+\abs{a_2}^2+\abs{a_3}^2+\abs{a_4}^2+2(\abs{b_1}^2+\abs{b_2}^2)
\notag\\
&+&
\sqrt{(\abs{a_1}^2-\abs{a_2}^2)^2+4\abs{b_1}^2(\abs{a_1}+\abs{a_2})^2}
+
\sqrt{(\abs{a_3}^2-\abs{a_4}^2)^2+4\abs{b_2}^2(\abs{a_3}+\abs{a_4})^2}
\bigg)
\notag\\
&\le &
{1\over2}\bigg( 
\abs{a_1}^2+\abs{a_2}^2+\abs{a_3}^2+\abs{a_4}^2+2(\abs{b_1}^2+\abs{b_2}^2)
\notag\\
&+&
\sqrt{\bigg((\abs{a_1}+\abs{a_2})^2+(\abs{a_3}+\abs{a_4})^2\bigg)\bigg((\abs{a_1}-\abs{a_2})^2+(\abs{a_3}-\abs{a_4})^2+4(\abs{b_1}^2+\abs{b_2}^2)\bigg)}
\bigg)
\notag\\
&:=&
{1\over2}\bigg(x+\sqrt{y(2x-y)}\bigg)
\notag\\
&\le &
x
\notag\\
&\le &
1/2,
\end{eqnarray}
where $x=\abs{a_1}^2+\abs{a_2}^2+\abs{a_3}^2+\abs{a_4}^2+2(\abs{b_1}^2+\abs{b_2}^2)$. The second inequality in \eqref{eq:case4} follows from Lemma \ref{le:bi} in which we have set $x=\sqrt{(\abs{a_1}^2-\abs{a_2}^2)^2+4\abs{b_1}^2(\abs{a_1}+\abs{a_2})^2}
$, $y=
\sqrt{(\abs{a_3}^2-\abs{a_4}^2)^2+4\abs{b_2}^2(\abs{a_3}+\abs{a_4})^2}$, $a=(a_3+a_4)^2$ and $b=(a_1+a_2)^2$. The last inequality in \eqref{eq:case4} follows from \eqref{eq:non2}. So Conjecture \ref{cj:main} holds.

Case 5. $\l$ is the eigenvalue of some $Y_j$, and $\m$ is the eigenvalue of some $Z_j$. Without loss of generality we may assume that $j=1$. 
Eqs. \eqref{eq:non1} and \eqref{eq:non2} imply that
$
\abs{b_1}^2
=
\abs{b_2+b_3+b_4}^2
\le
3(\abs{b_2}^2+\abs{b_3}^2+\abs{b_4}^2)
=
{3\over4}-3\abs{b_1}^2-3\sum^4_{j=1} \abs{a_j}^2.
$
Hence
\begin{equation}
\label{eq:non5}
\abs{b_1}^2\le {3\over16} - {3\over4}	\sum^4_{j=1} \abs{a_j}^2.
\end{equation}
On the other hand, by computation one can show that
\begin{eqnarray}
\label{eq:case5}
\l+\m
&=&
{1\over2}\bigg( 
\abs{a_1}^2+\abs{a_2}^2+\abs{a_3}^2+\abs{a_4}^2+4\abs{b_1}^2
\notag\\
&+&\sqrt{(\abs{a_1}^2-\abs{a_2}^2)^2+4\abs{b_1a_2^*+a_1b_1^*}^2}
+\sqrt{(\abs{a_3}^2-\abs{a_4}^2)^2+4\abs{b_1a_4^*+a_3b_1^*}^2}
\bigg)
\notag\\
&\le &
{1\over2}\bigg( 
\abs{a_1}^2+\abs{a_2}^2+\abs{a_3}^2+\abs{a_4}^2+4\abs{b_1}^2
\notag\\
&+&
\sqrt{(\abs{a_1}^2-\abs{a_2}^2)^2+4\abs{b_1}^2(\abs{a_1}+\abs{a_2})^2}
+
\sqrt{(\abs{a_3}^2-\abs{a_4}^2)^2+4\abs{b_1}^2(\abs{a_3}+\abs{a_4})^2}
\bigg)
\notag\\
&:=&
\L(\abs{a_1},\abs{a_2},\abs{a_3},\abs{a_4},\abs{b_1}^2).
\end{eqnarray}
It monotonically increases with $\abs{b_1}^2$. Using \eqref{eq:non2} we may assume that $\abs{a_1}=x \cos d \cos g$, $\abs{a_2}=x \cos d \sin g$, $\abs{a_3}=x \sin d \cos h$, and $\abs{a_4}=x \sin d \sin h$ where the real numbers $x\in[0, 1/2]$, and $d,g,h\in[0,\pi/2]$. 
Eqs. \eqref{eq:non5} and \eqref{eq:case5} imply that 
\begin{eqnarray}
\label{eq:a1a2}	
\l+\m &\le & \L(\abs{a_1},\abs{a_2},\abs{a_3},\abs{a_4},{3\over16} - {3\over4}	\sum^4_{j=1} \abs{a_j}^2)
\notag\\
&=&
{1\over8}\bigg(
3-8x^2
+2x\cos d\sqrt{f_1(d,x,g)}
+2x\sin d\sqrt{f_2(d,x,h)}
\bigg)
\end{eqnarray}
where
\begin{eqnarray}
\label{eq:f1dxg}
f_1(d,x,g)=
3-10x^2+2x^2\cos 2d
+(3-12x^2)\sin 2g	
+(-2x^2-2x^2\cos 2d) \sin^2 2g,
\end{eqnarray}
and
\begin{eqnarray}
\label{eq:f2dxg}
f_2(d,x,h)=
3-10x^2-2x^2\cos 2d
+(3-12x^2)\sin 2h	
+(-2x^2+2x^2\cos 2d) \sin^2 2h.
\end{eqnarray}
One can verify that $f_1(d,x,g)=f_2(\p/2-d,x,g)$. The last equation of \eqref{eq:a1a2} is unchanged under the switch of $d$ and $\p/2-d$, and the switch of $g$ and $h$ at the same time. So the maximum of \eqref{eq:a1a2} is achieved when $x\in[0,1/2]$, $d\in[0,\p/4]$, and $g,h\in[0,\pi/2]$. 

To prove the assertion, one need to obtain the maximum of \eqref{eq:a1a2}. For this purpose we need to obtain the maximum of the function $f_1$ in terms of $g$, and the maximum of the function $f_2$ in terms of $h$. The two functions $f_1$ and $f_2$ are respectively parabolas of cartesian coordinates $(\sin2g,f_1)$ and $(\sin2h,f_2)$. The axises of symmetry of $f_1$ and $f_2$ are respectively $\sin2g={3-12x^2\over4x^2+4x^2\cos2d}$ and $\sin2h={3-12x^2\over4x^2-4x^2\cos2d}$. 
If $\sin2g=1$ or $\sin2h=1$, then we respectively obtain $x={1\over\sqrt{4+{8\over3}\cos^2 d}}$ or $x={1\over\sqrt{4+{8\over3}\sin^2 d}}$. We discuss three subcases in terms of the above facts, $\sin2g\le1$, $\sin2h\le1$ and $\cos d \ge \sin d$.

Subcase 5.1.  $x\in[0,{1\over\sqrt{4+{8\over3}\cos^2 d}}]$. One can show that $\max_g f_1(d,x,g) = f_1(d,x,\p/4)$ and $\max_h f_2(d,x,h) = f_1(d,x,\p/4)$. Then one can show that \eqref{eq:a1a2} is upper bounded by $1/2$.

Subcase 5.2.  $x\in[{1\over\sqrt{4+{8\over3}\cos^2 d}}, {1\over\sqrt{4+{8\over3}\sin^2 d}}]$. One can show that $\max_g f_1(d,x,g)$ is achieved when $\sin 2g= {3-12x^2\over4x^2+4x^2\cos2d}$, and $\max_h f_2(d,x,h) = f_1(d,x,\p/4)$. Then one can show that \eqref{eq:a1a2} is upper bounded by $3/8$.

Subcase 5.3.  $x\in[{1\over\sqrt{4+{8\over3}\sin^2 d}},{1\over2}]$. One can show that $\max_g f_1(d,x,g)$ is achieved when $\sin 2g= {3-12x^2\over4x^2+4x^2\cos2d}$, and $\max_h f_2(d,x,h)$ is achieved when $\sin 2h= {3-12x^2\over4x^2-4x^2\cos2d}$. Then one can show that \eqref{eq:a1a2} is upper bounded by $3/8$. 

We have shown that \eqref{eq:a1a2} is upper bounded by $1/2$, i.e., $\l+\m\le1/2$. So Conjecture \ref{cj:main} holds in Case 5.

To conclude we have proved Conjecture \ref{cj:main} for the matrices $A,B$ in all five cases for $\l,\m$. This completes the proof.	
\end{proof}

Using the statement of Lemma \ref{le:cj} (iii), we may assume in Conjecture \ref{cj:main} that $A$ or $B$ is diagonal if and only if it is normal. Hence, Proposition \ref{pp:nonnormal+normal} implies that Conjecture \ref{cj:main} holds when $X\in\cP$ where one of $A$ and $B$ is normal.

\section{Conjecture \ref{cj:main} with non-normal matrices $X$: family 2}
\label{sec:nonnormal-sy}

In this section we investigate Conjecture \ref{cj:main} with $X$ defined as follows.

\begin{definition}
\label{df:nonnormal}	
Let $A=D_AP_A,B=D_BP_B$, where $D_A=\diag(a_1,a_2,\cdots,a_d)$, $D_B=\diag(b_1,b_2,\cdots,b_d)$ and $P_A,P_B$ are permutation matrices with zero-diagonals which make $A,B$ satisfy the first equation of \eqref{eq:tra} naturally. Meanwhile, $A$ and $B$ are under the following constraint 
\begin{equation}
\begin{aligned}
\label{eq:transformtra1}
\sum_{i=1}^d
(\abs{a_i}^2+\abs{b_i}^2)=\frac{1}{d}
\end{aligned}
\end{equation}
which follows from the second equation of \eqref{eq:tra}. 
\qed
\end{definition}
Such $X$ may be normal or non-normal, and may have dimension $d\ge4$. It is different from the set $\cP$ in Definition \ref{df:cp}. We shall prove Conjecture \ref{cj:main} when $A,B$ satisfy Definition \ref{df:nonnormal}, and two additional conditions respectively in Lemma \ref{le:ra1}, \ref{le:ada}.
In particular, Lemma \ref{le:ada} proves Conjecture \ref{cj:main} for $d\geq 5$. 

\begin{lemma}
\label{le:ra1}
Eq. \eqref{eq3} holds when $A,B$ satisfy Definition \ref{df:nonnormal} and $\rank A=1$.
\end{lemma}
\begin{proof}
For $\rank A=1$, it is safe to fix the only nonzero element $a$ of $A$ in the second entry of the first line of $A$, since there exists proper permutation matrix $P$ such that the only nonzero entry of matrix $PAP^\dg$ is the second entry of the first line for any $A$ with $\rank=1$. Let $B=D_BP_B$, where $D_B=\diag(b_1,\cdots,b_d)$ and $P_B$ is a permutation matrix with zero-diagonal. By calculation we have $B^\dg B=P^\dg_BD^\dg_BD_BP_B=P^\dg_B
\bma
\abs{b_1}^2 &   &  \\
             &\ddots & \\
             &   & \abs{b_d}^2
\ema P_B=\diag(\abs{b_{i_1}}^2,\cdots,\abs{b_{i_d}}^2)$, where $(i_1,\cdots,i_d)=\sigma(1,\cdots,d)$ with the constraint $i_j\neq j $, $\forall j\in\{1,2,\cdots,d\}$ and $\sigma$ here is a permutation. Then, one can show $X^\dg X=Y_1\oplus Y_2$, where  
$
Y_1=\bma
B^\dg B & aB^\dg \\
\bar{a}B& \abs{a}^2I+B^\dg B
\ema$ and 
$Y_2=\bma
B^\dg B & \cdots & 0  \\
\vdots  & \ddots & \vdots  \\
0 & \cdots & B^\dg B
\ema$. By calculation we obtain the eigenpolynomial of $X^\dg X$ as follow
\begin{equation}
\begin{aligned}
\label{eigenp:x}
\abs{\lambda I-X^\dg X}=[\prod_{j=1}^d(\lambda-\abs{b_j}^2)]^{d-2}[\prod_{j=1}^d(\lambda^2-(\abs{a}^2+\abs{b_j}^2+\abs{b_{i_j}}^2)\lambda+\abs{b_j}^2\abs{b_{i_j}}^2)].
\end{aligned}
\end{equation}

One can show that $\lambda^2-(\abs{a}^2+\abs{b_j}^2+\abs{b_{i_j}}^2)\lambda+\abs{b_j}^2\abs{b_{i_j}}^2\leq 0$ when $\lambda=\abs{b_j}^2$, which implies the largest two eigenvalues of $X^\dg X$ must be the roots of the second product of \eqref{eigenp:x}. Hence, the sum of the largest two eigenvalues in \eqref{eigenp:x} can be expressed as follow
\begin{equation}
\begin{aligned}
\label{eq:eigensum1}
\lambda_1+\lambda_2=\max_{j\neq k}(\frac{(\abs{a}^2+\abs{b_j}^2+\abs{b_{i_j}}^2+\sqrt{\delta_j})+(\abs{a}^2+\abs{b_k}^2+\abs{b_{i_k}}^2+\sqrt{\delta_k})}{2}),
\end{aligned} 
\end{equation}
where $\delta_j=(\abs{a}^2+\abs{b_j}^2+\abs{b_{i_j}}^2)^2-4\abs{b_j}^2\abs{b_{i_j}}^2$ and $\delta_k=(\abs{a}^2+\abs{b_k}^2+\abs{b_{i_k}}^2)^2-4\abs{b_k}^2\abs{b_{i_k}}^2$. We have $\sqrt{\delta_j}\leq\frac{1}{d}$ and $\sqrt{\delta_k}\leq\frac{1}{d}$, $\forall j,k$ which follow from Eq. \eqref{eq:transformtra1}. Hence, Eq. \eqref{eq:eigensum1} implies the sum of the largest two  eigenvalues of $X^\dg X$ is at most $\frac{2}{d}$ which isn't greater than $\frac{1}{2}$ for $d\geq  4$. This completes the proof.  
\end{proof}

Using the statement of Lemma \ref{le:cj} (iv), we can make the same conclusion if $\rank B=1$. Hence, Eq. \eqref{eq3} holds when $A,B$ satisfy Definition \ref{df:nonnormal} and one of them has rank one. 

We have seen that it is not easy to characterize the eigenpolynomial of $X^\dg X$. In the following lemma, we use Gershgorin circle theorem and Brauer theorem to study Conjecture \ref{cj:main}.
They are two important theorems in the field of localization of eigenvalues to localize the largest two eigenvalues. The following fact will be used in the proof of Lemma \ref{le:ada}.

It follows from \eqref{eq:x} that $X^\dg X=H_1+H_2$ where
\begin{eqnarray}
\label{eq:x2}
H_1&:=&A^\dg A\otimes I+I\otimes B^\dg B,
\notag\\
H_2&:=&A^\dg\otimes B+A\otimes B^\dg.
\end{eqnarray}
Furthermore, the first equation of \eqref{eq:tra} implies that $\tr H_2=0$, and the second equation of \eqref{eq:tra} implies that $\tr X^\dg X=\tr H_1=\sum^{d^2}_{j=1}\sigma_j^2=1$. So $X^\dg X$ can be regarded as a normalized quantum state in terms of quantum physics.

\begin{lemma}
\label{le:ada}
Eq. \eqref{eq3} holds for $d\geq 5$ when $A$ and $B$ satisfy Definition \ref{df:nonnormal}.
\end{lemma}
\begin{proof}
$H_1$ in Eq. \eqref{eq:x2} is diagonal and $H_2$ in Eq. \eqref{eq:x2} is a Hermitian matrix with zero-diagonal when $A$ and $B$ satisfy Definition \ref{df:nonnormal}. There exist two permutations $\sigma$ and $\tau$ with $\sigma(k)\neq k,\tau(k)\neq k,\forall k\in\{1,\cdots,d\}$ which are respectively equivalent to $P_A$ and $P_B$. We find that the $(k,\sigma(k))$ entry of $A$ is $a_k$ and the $(k,\tau(k))$ entry of $B$ is $b_k$. So the $(\sigma(k),k)$ entry of $A^\dg$ is $a_k^*$ and the $(\tau(k),k)$ entry of $B^\dg$ is $b_k^*$. They imply that exactly two entries in each row of $H_2$ can be expressed with $a_i,b_j$ and their conjugates, for $i,j\in\{1,\cdots,d\}$. It implies that the two elements $a_{\sigma^{-1}(i)}^*b_j$ which is the $\big(d(i-1)+j,d(\sigma^{-1}(i)-1)+\tau(j)\big)$ entry of $X^\dg X$ and $a_ib_{\tau^{-1}(j)}^*$ which is the $\big(d(i-1)+j,d(\sigma(i)-1)+\tau^{-1}(j)\big)$ entry of $X^\dg X$ are both in the $(d(i-1)+j)$'th row of $X^\dg X$ and also are non-diagonal entries of $X^\dg X$. Further the diagonal entry of $X^\dg X$ in this row is $(\abs{a_{\sigma^{-1}(i)}}^2+\abs{b_{\tau^{-1}(j)}}^2)$. Recall that $\sigma(i)\neq i,\tau(i)\neq i,\forall i\in\{1,\cdots,d\}$. Eq. \eqref{eq:spec11} implies that the largest eigenvalue $\lambda_1$ of $X^\dg X$  satisfies 
\begin{eqnarray}
\label{eq:eigensum21}
\lambda_1
&\leq &\max\limits_{i,j}\big(\abs{a_{\sigma^{-1}(i)}}^2+\abs{b_{\tau^{-1}(j)}}^2+\abs{a_{\sigma^{-1}(i)}}\abs{b_j}+\abs{a_i}\abs{b_{\tau^{-1}(j)}}\big),
\end{eqnarray}
 and the fact $\sigma(k)\neq k,\tau(k)\neq k,\forall k\in\{1,\cdots,d\}$. Applying the basic inequality, we obtain $\abs{a_p}\abs{b_q}+\abs{a_s}\abs{b_t}\leq (\abs{a_p}^2+\abs{b_q}^2+\abs{a_s}^2+\abs{b_t}^2)/2$. Since $p\neq s$ and $q\neq t$, the constraint \eqref{eq:transformtra1} implies that $\abs{a_p}\abs{b_q}+\abs{a_s}\abs{b_t}\leq \frac{1}{2d}$. The constraint \eqref{eq:transformtra1} also implies that $\abs{a_i}^2+\abs{b_j}^2\leq\frac{1}{d}$. Hence, we have $\lambda_1\leq\frac{3}{2d}$ and thus $\lambda_1+\lambda_2\leq\frac{3}{d}$. It implies Conjecture \ref{cj:main} holds for $d\geq 6$. Next we will prove Conjecture \ref{cj:main} holds for $d=5$.

Let's recall \eqref{eq:x2}. Eq. \eqref{eq:eivaineq1} implies the second largest eigenvalue of $X^\dg X$ satisfy $\lambda_2(X^\dg X)\leq\lambda_2(H_1)+\lambda_1(H_2)$, where $\lambda_2(H_1)$ means the second largest eigenvalue of $H_1$ and $\lambda_1(H_2)$ means the largest eigenvlaue of $H_2$. We find $\lambda_2(H_1)$ is the second largest diagonal element of $H_1$. Applying Lemma \ref{le:ovalth} to $H_2$, one can show that $\lambda_1(H_2)$ should not greater than the largest root of these quadratic polynomials $\lambda^2=\sum\limits_{k=1\atop k\neq k_1}^n\abs{X^\dg X_{(k_1,k)}}\sum\limits_{k=1\atop k\neq k_2}^n\abs{X^\dg X_{(k_2,k)}}$ for $k_1\neq k_2$. Suppose $k_1=d(i-1)+j$ and $k_2=d(p-1)+q$ with $(i,j)\neq(p,q)$. Hence, we can bound $\lambda_1(X^\dg X)+\lambda_2(X^\dg X)$ as follow.
\begin{equation}
\begin{aligned}
\label{eq:eigensum31}
(\lambda_1+\lambda_2)\leq\max\limits_{(i,j)\neq(p,q)}\Bigg(\abs{a_{\sigma^{-1}(i)}}^2+\abs{b_{\tau^{-1}(j)}}^2+\abs{a_{\sigma^{-1}(p)}}^2+\abs{b_{\tau^{-1}(q)}}^2
+2\sqrt{c}\Bigg),
\end{aligned}
\end{equation}
where $c=\big(\abs{a_{\sigma^{-1}(i)}}\abs{b_j}+\abs{a_i}\abs{b_{\tau^{-1}(j)}}\big)\big(\abs{a_{\sigma^{-1}(p)}}\abs{b_q}+\abs{a_p}\abs{b_{\tau^{-1}(q)}}\big)$.

Suppose $\abs{a_i}$ and $\abs{b_j}$ be in the decreasing order. In order to obtain the upper bound of Eq. \ref{eq:eigensum31}, it is safe to let $x_1=\abs{a_1}$, $x_2=\abs{a_2}$, $x_3=\abs{b_1}$ and $x_4=\abs{b_2}$ and other $\abs{a_i}$ and $\abs{b_j}$ all equal zero. Then our problem can be transformed into an optimization task as follow.
\begin{equation}
\begin{aligned}
\label{eq:opt1}
\max &\quad f(x_1,x_2,x_3,x_4)=2x_1^2+x_3^2+x_4^2+2\sqrt{(x_1x_3+x_2x_4)(x_1x_4+x_2x_3)}\\
s.t. &\quad x_i>0,\quad i=1,2,3,4,  \\
     &\quad \sum_{i=1}^4x_i^2=\frac{1}{d},\\
     &x_1-x_2\geq 0,\quad x_3-x_4\geq 0, \\
     &x_1^2-x_2^2\geq x_3^2-x_4^2.
\end{aligned}
\end{equation}

Since $\sum\limits_{i=1}^4x_i^2=\frac{1}{d}$, it is safe to let $x_1=\sqrt{\frac{1}{d}}(\cos a\cos b)$, $x_2=\sqrt{\frac{1}{d}}(\cos a\sin b)$, $x_3=\sqrt{\frac{1}{d}}(\sin a\cos c)$ and $x_4=\sqrt{\frac{1}{d}}(\sin a\sin c)$. Then, we obtain $f=\frac{1}{4d}\big(4+\cos2(a-b)+2\cos2b+\cos2(a+b)+2\sqrt{2}\sqrt{\sin^22a(\sin2b+\sin2c)}\big)$, which implies $f$ obtains its maximum only when $\sin2c=1$. Furthermore, $\sin2c=1$ implies $x_3=x_4$. Substituting $x_3=x_4$ and $\sum\limits_{i=1}^4x_i^2=\frac{1}{d}$ into $f(x_1,x_2,x_3,x_4)=2x_1^2+x_3^2+x_4^2+2\sqrt{(x_1x_3+x_2x_4)(x_1x_4+x_2x_3)}$, we can transform $f$ into a function with two variables, that is $f(x_1,x_2)=\frac{1}{d}+x_1^2-x_2^2+2(x_1+x_2)\sqrt{\frac{1}{2d}-\frac{1}{2}(x_1^2+x_2^2)}$. Let $d=5$ and the numberical result shows that even though there is no constraint $x_1\geq x_2$, $f$ is also upper bounded by $\frac{1}{2}$. 

To conclude all dimensions $d\geq5$ have been studied. This completes the proof.
\end{proof}

\section{Conjecture \ref{cj:main} with normal $X$ and $d\ne4$}
\label{sec:normal}

Reference \cite{pphh2010} investigated Conjecture \ref{cj:main} for normal matrices with $d=4$, as we have introduced in Lemma \ref{le:d=4}. In this section we extend Lemma \ref{le:d=4} to higher dimensions, so that Conjecture \ref{cj:main} is of more mathematical interest apart from its physical connection to the distillability problem.
\begin{lemma}
\label{le:d>4}
Lemma \ref{le:d=4} holds when $\cN_4$ is replaced by $\cN_d$ with $d>4$. 
\end{lemma}
\begin{proof}
The part in the proof of Lemma \ref{le:d=4} from the beginning to \eqref{eq13} applies here. 
This part applies to any $d>4$, and Eq. \eqref{eq13} is the first place in the proof of Lemma \ref{le:d=4} in which $d=4$ appears. Based on these facts, we begin our proof with $d>4$. According to \eqref{eq13}, we have
\begin{equation}
\begin{aligned}
\label{eq17}
\abs{a_1+b_1}^2+\abs{a_2+b_2}^2&\leq 2(\abs{a_1}^2+\abs{b_1}^2+\abs{a_2}^2+\abs{b_2}^2)
\leq\frac{2}{d}<\frac{1}{2}.
\end{aligned}
\end{equation}
Proposition \ref{pp:ab} implies that 
\eqref{eq15} is satisfied. So we have
\begin{equation}
\begin{aligned}
\label{eq18} 
\abs{a_1+b_1}^2+\abs{a_1+b_2}^2\leq\frac{3d-4}{d^2}<\frac{1}{2}.
\end{aligned}
\end{equation}
The two inequalities in \eqref{eq10} and \eqref{eq11} are saturated. So Lemma \ref{le:d=4} holds for $d\geq 4$.
\end{proof}

Next we show that Lemma \ref{le:d=4} no longer holds when $d=4$ is replaced by $d=3$. 

\begin{lemma}
\label{le:d=3}
Let $\chi_d$ be a subset of normal operators $X$ in \eqref{eq:x} satisfying constraints \eqref{eq:tra}. Then for $d=3$, we have 
\begin{equation}
\begin{aligned}
\label{eq19}
\frac{5}{9}\leq\sup_{X\in\chi_d}(\sigma_1^2+\sigma_2^2)\leq\frac{2}{3}
\end{aligned}
\end{equation}
where $\sigma_1$ and $\sigma_2$ are the two largest singular values of operator X.
\end{lemma}
\begin{proof}
According to \eqref{eq12}, we have
\begin{equation}
\begin{aligned}
\label{eq:tra0}
\abs{a_1+b_1}^2+\abs{a_2+b_2}^2&\leq 2(\abs{a_1}^2+\abs{b_1}^2+\abs{a_2}^2+\abs{b_2}^2)
\leq\frac{2}{d}=\frac{2}{3}.
\end{aligned}
\end{equation}
Proposition \ref{pp:ab} implies that
\begin{equation}
\begin{aligned}
\label{eq:tra1}
\max(\abs{a_1+b_1}^2+\abs{a_1+b_2}^2)=\frac{3d-4}{d^2}=\frac{5}{9}.
\end{aligned}
\end{equation}
Due to the relationship \eqref{eq9}, we obtain \eqref{eq19} for $d=3$.
\end{proof}

Let $a_1=\frac{4}{3\sqrt{10}}$, $a_2=\frac{-2}{3\sqrt{10}}$, $a_3=\frac{-2}{3\sqrt{10}}$ and $b_1=\frac{1}{3\sqrt{10}}$, $b_2=\frac{1}{3\sqrt{10}}$, $b_3=\frac{-2}{3\sqrt{10}}$. They satisfy \eqref{eq:tra} and saturate the first inequality in \eqref{eq19}.

\section*{Acknowledgments}

YS and LC were supported by Beijing Natural Science Foundation (4173076), the NNSF of China (Grant No. 11501024), and the Fundamental Research Funds for the Central Universities (Grant Nos. KG12001101, ZG216S1760 and ZG226S17J6). 

\appendix

\section{mathematical description of the distillability problem}
\label{app:dis}

Let $\cH=\cH_A\ox\cH_B$ be the bipartite Hilbert space with $\dim\cH_A=M$ and $\dim\cH_B=N$. In quantum physics,
the quantum state is a positive semidefinite matrix. For the sake of normalization in quantum physics, it is required that every quantum state be a unit vector. However the requirement does not play the essential in the distillability problem, and often causes  inconvenience in the mathematical expressions and discussion. In this paper, unless stated otherwise, the states will not
be normalized.
  
We shall work with the quantum state $\rho$ on $\cH$. Such $\rho$ is called a bipartite state of system $A$ and $B$. We have $\rho=\sum^M_{i,j=1}E_{ij}\ox \rho_{ij}$, where $E_{ij}$ is an $M\times M$ matrix whose elements are all zero, except that the $(i,j)$ entry is one. The partial transpose of $\rho$ with respect to the system $A$ is defined as $\rho^{\G}:=\sum^M_{i,j=1}E_{ji}\ox\rho_{ij}$. We say that $\rho$ is positive partial transpose (PPT) if $\rho^{\Gamma}\ge0$. Otherwise $\rho$ is negative partial transpose (NPT), i.e., $\rho^{\Gamma}$ has at least one negative eigenvalue. The NPT states are entangled states due to the
Peres-Horodecki criterion in quantum information \cite{peres1996,hhh96}. We say that a quantum state is pure when it has rank one.

Since the distillability problem requires many copies of the same states, we need further the concept of composite system.
Let $\rho_{A_iB_i}$ be an $M_i\times N_i$ state of rank $r_i$ acting
on the Hilbert space $\cH_{A_i}\ox\cH_{B_i}$, $i=1,2$, with $\dim\cH_{A_i}=M_i$ and $\dim\cH_{B_i}=N_i$. Suppose
$\rho$ of systems $A_1,A_2$ and $B_1,B_2$ is a state on the
Hilbert space
$\cH_{A_1}\ox\cH_{B_1}\ox\cH_{A_2}\ox\cH_{B_2}$, such that the partial trace $\tr_{A_1B_1}\rho=\rho_{A_2B_2}$ and
$\tr_{A_2B_2}\r=\rho_{A_1B_1}$. By
switching the two middle factors, we can regard $\rho$ as a
\textit{composite} bipartite state on the Hilbert space
$\cH_A\ox\cH_B$ where $\cH_A=\cH_{A_1}\ox\cH_{A_2}$ and
$\cH_B=\cH_{B_1}\ox\cH_{B_2}$. We write
$\rho=\rho_{A_1A_2:B_1B_2}$. One can verify that $\rho$ is an $M_1M_2\times N_1N_2$
state of rank at most $r_1r_2$. For example the
\textit{tensor product} $\rho=\rho_{A_1B_1}\ox\rho_{A_2B_2}$ is an $M_1M_2\times N_1N_2$ state of rank $r_1r_2$.
The above definition can be generalized to the tensor product
of $N$ states $\rho_{A_iB_i},i=1,\ldots,N$. They form a bipartite
state on the Hilbert space
$\cH_{A_1,\cdots,A_N}\ox\cH_{B_1,\cdots,B_N}$. It is written as $\cH^{\ox n}$ with $\cH_{A_i}\ox\cH_{B_i}=\cH$. 

Third we shall refer to the notations $\ket{\ps}$ and $\bra{\ps}$ in quantum physics respectively as a column vector and its conjugate transpose in linear algebra. In quantum information, the well-known Werner state in $\bbC^d\ox\bbC^d$ is defined as ${I+\alpha \sum^{d}_{i,j=1}E_{ij}\ox E_{ji} \over d^2+\alpha d}$, where the real number $\alpha\in[-1,1]$ \cite{werner89}.
We introduce the definition of distillable states as follows \cite{dss00}.
\begin{definition}
 \label{def:distillation}
A bipartite state $\rho$ is {\em $n$-distillable} under local operations and classical communications if there
exists a Schmidt-rank-two bipartite state $\ket{\ps}\in\cH^{\ox n}$ such
that $\bra{\ps} ({\rho^{\ox n}})^\Gamma \ket{\ps}<0$. Otherwise we say
that $\rho$ is $n$-undistillable. We say that $\rho$ is {\em
distillable} if it is $n$-distillable for some $n\geq1$. 
\qed
\end{definition}
The definition shows that PPT states are not distillable. It has been shown \cite{dss00} that all NPT bipartite states can be locally converted into NPT Werner states. Using Definition \ref{def:distillation}, one can show that the distillability of NPT Werner states is equivalent to that with $\alpha=-1/2$. 
So the distillability problem indeed asks whether Werner states with $\alpha=-1/2$ are distillable. 

Now we can explain the special case proposed in \cite{pphh2010}. It means that Conjecture \ref{cj:main} with $d=4$ is equivalent to the 2-undistillability of Werner states in $\bbC^4\ox\bbC^4$ with $\alpha=-1/2$. Definition \ref{def:distillation} shows that if Conjecture \ref{cj:main} with $d=4$ was true, it is still possible that Werner states might be $n$-distillable with some integer $n>2$. However it is widely believed that Werner states with $\alpha=-1/2$ may be not distillable \cite{dcl00,pphh2010,dss00, Hiroshima2008Bound, Hiroshima2012A}.

\section{the proof of Lemma \ref{le:d=4}}
\label{app:lemma1}
\begin{proof}
First we can see the operator X of the form \eqref{eq:x} is normal iff operators A and B are normal. Since X is normal which means X is diagonalizable, then we can replace singular values with moduli of eigenvalues, which means
\begin{equation}
\begin{aligned}
\label{eq5}
\lambda_{ij}=a_i+b_j
\end{aligned}
\end{equation}
where $a_i$ and $b_j$ are eigenvalues of A and B, respectively, and $\lambda_{ij}$ are eigenvalues of X. We then have
\begin{equation}
\begin{aligned}
\label{eq6}
\sup_{X\in\chi_d}(\sigma_1^2+\sigma_2^2)&=\sup_{X\in\chi_d}(\abs{\lambda_1}^2+\abs{\lambda_2}^2)\\
&=\sup_{X\in\chi_d}\max_{i,j,k,l\in\{1,\cdots,d\},(i,j)\neq(k,l)}(\abs{a_i+b_j}^2+\abs{a_k+b_l}^2)
\end{aligned}
\end{equation}
where $\lambda_1$ and $\lambda_2$ are two eigenvalues of X with largest moduli. The constraints \eqref{eq:tra} on X imply the following constraints on $a_i$ and $b_j$
\begin{equation}
\begin{aligned}
\label{eq7}
\sum_{i=1}^da_i=\tr A=0,\quad \sum_{j=1}^db_j=\tr B=0,
\end{aligned}
\end{equation}
\begin{equation}
\begin{aligned}
\label{eq8}
\sum_{i=1}^d\abs{a_i}^2+\sum_{j=1}^d\abs{b_j}^2=\tr A^\dagger A+\tr B^\dagger B=\frac{1}{d}.
\end{aligned}
\end{equation}
Considering such two pairs of number $(i,j)$ and $(k,l)$ with the constraint $(i,j)\neq(k,l)$, we get the following two cases:
\begin{displaymath}
(1)\quad (i\neq k)\wedge (j\neq l);
\end{displaymath}
\begin{displaymath}
(2)\quad ((i=k)\wedge (j\neq l))\vee ((i\neq k)\wedge(j=l)).
\end{displaymath}
Due to the alternative property, we can only consider the left term or the right term of $\vee$ in the second case. Then we can go further with the \eqref{eq6} as follows.
\begin{equation}
\begin{aligned}
\label{eq9}
\sup_{X\in\chi_d}(\sigma_1^2+\sigma_2^2)&=\sup_{X\in\chi_d}\max\{\abs{a_1+b_1}^2+\abs{a_2+b_2}^2,\quad \abs{a_1+b_1}^2+\abs{a_1+b_2}^2\}.
\end{aligned}
\end{equation}
Thus, to prove the Lemma we have to show that the following inequalities hold:
\begin{equation}
\begin{aligned}
\label{eq10}
\abs{a_1+b_1}^2+\abs{a_2+b_2}^2\leq\frac{1}{2}
\end{aligned}
\end{equation}
\begin{equation}
\begin{aligned}
\label{eq11}
\abs{a_1+b_1}^2+\abs{a_1+b_2}^2\leq\frac{1}{2}
\end{aligned}
\end{equation}
under the constraints \eqref{eq7} and \eqref{eq8} with d=4. The first inequality comes easily from the identity
\begin{equation}
\begin{aligned}
\label{eq12}
\abs{x+y}^2=2(\abs{x}^2+\abs{y}^2)-\abs{x-y}^2\leq 2(\abs{x}^2+\abs{y}^2)
\end{aligned}
\end{equation}
which implies
\begin{equation}
\begin{aligned}
\label{eq13}
\abs{a_1+b_1}^2+\abs{a_2+b_2}^2&\leq 2(\abs{a_1}^2+\abs{b_1}^2+\abs{a_2}^2+\abs{b_2}^2)\\
&\leq\frac{2}{d}=\frac{1}{2}.
\end{aligned}
\end{equation}
Then the next work is to show the inequality \eqref{eq11} holds with $d=4$. It follows from the following Proposition \ref{pp:ab} proven by \cite{pphh2010}. 
\end{proof}

\begin{proposition}
\label{pp:ab}
Suppose $\overrightharp{a}$ and $\overrightharp{b}$ are $d\geq 3$ dimensional vectors with complex elements $\tilde{a_i}$ and $\tilde{b_i}$ satisfying the constraints
\begin{equation}
\begin{aligned}
\label{eq14}
\sum_{i=1}^d\tilde{a_i}=\sum_{i=1}^d\tilde{b_i}=0,
~~~~~
\quad \sum_{i=1}^d\abs{\tilde{a_i}}^2+\sum_{i=1}^d\abs{\tilde{b_i}}^2=\frac{1}{d}.
\end{aligned}
\end{equation}
Then the following equality holds,
\begin{equation}
\begin{aligned}
\label{eq15}
\max_{\overrightharp{a},\overrightharp{b}}(\abs{\tilde{a_1}+\tilde{b_1}}^2+\abs{\tilde{a_1}+\tilde{b_2}}^2)=\frac{3d-4}{d^2}.
\end{aligned}
\end{equation}
\qed
\end{proposition}

\renewcommand\refname{References}
\bibliographystyle{ieeetr} 
\bibliography{dis}

\begin{thebibliography}{10}

\bibitem{bernstein2009matrix}
D.~Bernstein, {\em Matrix Mathematics: Theory, Facts, and Formulas (Second
  Edition)}.
\newblock Princeton University Press, 2009.

\bibitem{pphh2010}
M.~H. Łukasz Pankowski, Marco~Piani and P.~Horodecki, ``A few steps more
  towards npt bound entanglement,'' {\em IEEE. Trans. Inf. Theory}, vol.~56,
  pp.~1--19, Aug 2010.

\bibitem{dss00}
D.~P. DiVincenzo, P.~W. Shor, J.~A. Smolin, B.~M. Terhal, and A.~V. Thapliyal,
  ``Evidence for bound entangled states with negative partial transpose,'' {\em
  Phys. Rev. A}, vol.~61, p.~062312, May 2000.

\bibitem{Terhal2000A}
B.~M. Terhal, ``A family of indecomposable positive linear maps based on
  entangled quantum states,'' {\em Linear Algebra and Its Applications},
  vol.~323, no.~1, pp.~61--73, 2000.

\bibitem{Poon2015Preservers}
E.~Poon, ``Preservers of maximally entangled states,'' {\em Linear Algebra and
  Its Applications}, vol.~468, no.~468, pp.~122--144, 2015.

\bibitem{Cariello2016A}
D.~Cariello, ``A gap for ppt entanglement,'' {\em Linear Algebra and Its
  Applications}, vol.~529, 2016.

\bibitem{Hou2013Linear}
J.~Hou and X.~Qi, ``Linear maps preserving separability of pure states,'' {\em
  Linear Algebra and Its Applications}, vol.~439, no.~5, pp.~1245--1257, 2013.

\bibitem{Alfsen2012Finding}
E.~Alfsen and F.~Shultz, ``Finding decompositions of a class of separable
  states,'' {\em Linear Algebra and Its Applications}, vol.~437, no.~10,
  pp.~2613--2629, 2012.

\bibitem{Woerdeman2004The}
H.~J. Woerdeman, ``The separability problem and normal completions,'' {\em
  Linear Algebra and Its Applications}, vol.~376, no.~1, pp.~85--95, 2004.

\bibitem{dcl00}
W.~D\"ur, J.~I. Cirac, M.~Lewenstein, and D.~Bru\ss{}, ``Distillability and
  partial transposition in bipartite systems,'' {\em Phys. Rev. A}, vol.~61,
  p.~062313, May 2000.

\bibitem{vd06}
R.~O. Vianna and A.~C. Doherty, ``Distillability of werner states using
  entanglement witnesses and robust semidefinite programs,'' {\em Phys. Rev.
  A}, vol.~74, p.~052306, Nov 2006.

\bibitem{br03}
S.~Bandyopadhyay and V.~Roychowdhury, ``Classes of $n$-copy undistillable
  quantum states with negative partial transposition,'' {\em Phys. Rev. A},
  vol.~68, p.~022319, Aug 2003.

\bibitem{cc08}
L.~Chen and Y.-X. Chen, ``Rank-three bipartite entangled states are
  distillable,'' {\em Phys. Rev. A}, vol.~78, p.~022318, Aug 2008.

\bibitem{hhh97}
M.~Horodecki, P.~Horodecki, and R.~Horodecki, ``Inseparable two spin-
  $\frac{1}{2}$ density matrices can be distilled to a singlet form,'' {\em
  Phys. Rev. Lett.}, vol.~78, pp.~574--577, Jan 1997.

\bibitem{rains1999}
E.~M. Rains, ``Bound on distillable entanglement,'' {\em Phys. Rev. A},
  vol.~60, p.~179, 1999.

\bibitem{cd11jpa}
L.~Chen and D.~Z. Djokovic‡, ``Distillability and ppt entanglement of low-rank
  quantum states,'' {\em Journal of Physics A: Mathematical and Theoretical},
  vol.~44, no.~28, p.~285303, 2011.

\bibitem{cd16pra}
L.~Chen and D.~Z. Djokovic‡, ``Distillability of non-positive-partial-transpose
  bipartite quantum states of rank four,'' {\em Phys. Rev. A}, vol.~94,
  p.~052318, Nov 2016.

\bibitem{hj85}
R.~Horn and C.~Johnson, {\em {Matrix Analysis}}.
\newblock Cambridge, England: Cambridge University Press, 1985.

\bibitem{peres1996}
A.~Peres, ``Separability criterion for density matrices,'' {\em Phys. Rev.
  Lett.}, vol.~77, p.~1413, 1996.

\bibitem{hhh96}
M.~{Horodecki}, P.~{Horodecki}, and R.~{Horodecki}, ``{Separability of mixed
  states: necessary and sufficient conditions},'' {\em Physics Letters A},
  vol.~223, pp.~1--8, Feb. 1996.

\bibitem{werner89}
R.~F. Werner, ``Quantum states with einstein-podolsky-rosen correlations
  admitting a hidden-variable model,'' {\em Physical Review A}, vol.~40,
  pp.~4277--4281, Oct 1989.

\bibitem{Hiroshima2008Bound}
T.~Hiroshima, ``Bound entangled states with non-positive partial transpose
  exist,'' {\em Physics}, 2008.

\bibitem{Hiroshima2012A}
T.~Hiroshima, ``A problem of existence of bound entangled states with
  non-positive partial transpose and the hilbert's 17th problem,'' {\em Eprint
  Arxiv}, 2012.

\end{thebibliography}
\end{document}